\documentclass[conference]{IEEEtran}

\usepackage{amsmath,amssymb}
\usepackage{enumitem}
\usepackage{amsthm}
\usepackage{algcompatible}
\usepackage{algorithm}
\usepackage{algpseudocode}
\usepackage[pdftex]{graphicx}
\usepackage{tabularx}
\usepackage{tikz}
\definecolor{cof}{RGB}{219,144,71}
\definecolor{pur}{RGB}{186,146,162}
\definecolor{greeo}{RGB}{91,173,69}
\definecolor{greet}{RGB}{52,111,72}
\usetikzlibrary{patterns}
\usetikzlibrary{arrows}
\usepackage{pgfplots}
\usepackage{multirow}
\usepackage{hhline}
\usepackage{bbm}
\usepackage{subcaption}
\makeatletter
\newcommand*{\rom}[1]{\expandafter\@slowromancap\romannumeral #1@}
\newcommand{\Mod}[1]{\ (\mathrm{mod}\ #1)}
\makeatother
 \newtheorem{lem}{Lemma}
  \newtheorem{thm}{Theorem}
\begin{document}
\title{Distributed Control and  Quality-of-Service in Multihop Wireless Networks}

\author{\IEEEauthorblockN{Ashok Krishnan K.S. and Vinod Sharma}
\IEEEauthorblockA{Dept. of ECE, Indian Institute of Science, Bangalore, India\\
%\\
%Indian Institute of Science\\
%Bangalore, India 560012\\
Email: \{ashok, vinod\}@ece.iisc.ernet.in
}}

\maketitle

\begin{abstract}
Control of wireless  multihop networks, while simultaneously meeting end-to-end mean delay requirements of different flows is a challenging problem. Additionally, distributed computation of control parameters adds to the complexity. Using the notion of discrete review used in fluid control of networks, a distributed algorithm is proposed for control of  multihop wireless networks with interference constraints. The algorithm meets end-to-end mean delay requirements  by solving an optimization problem at review instants. The optimization incorporates delay requirements as weights in the function being maximized. The weights are dynamic and vary  depending on queue length information. The optimization is done in a distributed manner using an incremental gradient ascent algorithm. The stability of the network under the proposed policy is analytically studied and the policy is shown to be throughput optimal.
\end{abstract}

% \fi
\IEEEpeerreviewmaketitle

\section{Introduction and Literature Review}
A multihop wireless network consists of nodes transmitting and receiving information over the wireless medium, with data from a source node having to pass through multiple hops before it can reach its destination. The control of wireless networks, involving scheduling, routing and power control, is a complex and challenging problem. Applications often require distributed control as opposed to centralized control. Distributed control algorithms may not always match centralized algorithms in terms of performance. However, they offer ease of implementation from a practical perspective in scenarios where it may not be feasible to collect information from all over the network before arriving at a control decision.\\
\indent Different flows in a network, arising from different applications, may ask for Quality-of-Service (QoS). QoS may vary depending on the nature of the application. Some applications require an end-to-end  mean delay guarantee on the packets being transmitted. Some others, such as a live streaming video, may require all packets to satisfy a \emph{hard} delay requirement. In some cases, the QoS constraint is a bandwidth requirement for the user. Services involving VoIP (Voice over IP) are sensitive to delay and delay variability in the network, and require preferential treatment over other packets \cite{markopoulou2002assessment}. Another service that requires QoS is remote health-care, which involves collection of data about a patient from a remote location and transmitting it elsewhere to be analysed \cite{shah2016remote}. Such applications in the context of the Internet of Things (IoT) \cite{atzori2010internet}  will require the coming together of different kinds of traffic with various QoS requirements, and with different levels of sensitivity \cite{awan2014modelling}.\\
\indent While directly solving problems that involve QoS requirements may not be straightforward, one can look for appropriate asymptotic solutions. One approach is to  study the network in the large queue length regime, and translate mean delay requirements of flows into \emph{effective bandwidth} and \emph{effective delay} as given by large deviations theory, and formulate these as physical layer requirements \cite{she2016energy}. In the case of multihop networks, however, owing to the complex coupling between queues, such a formulation is not easy to obtain \cite{LauSurvey}.\\
\indent Backpressure based methods are common in network control. These are connected to control based on Lyapunov Optimization \cite{GeorgBook}.  Backpressure based algorithms may not provide good delay performance, especially in lightly loaded conditions \cite{Cui}, \cite{sharma2007opportunistic}. Techniques based on Markov Decision processes (MDPs) are also popular \cite{Kumar}. In \cite{SatWiOPt}, \cite{SatVS}, the problem of minimizing power while providing mean  and hard delay guarantees is studied. However the algorithm requires knowledge of system statistics and is not throughput optimal.\\
\indent Fluid limits \cite{Meyn08} are an asymptotic technique used to study networks and obtain control solutions. The network parameters and variables, under a suitable scaling, are shown to converge to deterministic processes. This is called the \emph{fluid limit}. The stability of the fluid limit has a direct bearing upon the stability of the original stochastic system \cite{dai1995positive},  \cite{andrews2004scheduling}. The technique of discrete review  is used in \cite{Marg00}. Here, the network is reviewed at certain time instants, and control decisions are taken till the next time instant using the state information observed. In \cite{shroff13} the authors use fluid limit based techniques to establish the stability of a per-queue based scheduling algorithm. A robust fluid model, obtained by adding stochastic variability to the conventional fluid model, is discussed in \cite{Bert15}. Another algorithm using per hop queue length information, with a low complexity approximation that stabilizes a fraction of the capacity region is given in \cite{srikant}. However, the algorithm does not address delay QoS.\\
\indent Our main contributions in this paper are summarised below.
\begin{itemize}[noitemsep,topsep=0pt]
	\item We propose an algorithm that solves a weighted optimization in order to address mean delay requirements of different flows. The weights are time varying and state dependent, as opposed to fixed weight schemes. This assigns dynamic priorities to different flows.
	\item Our policy uses the technique of discrete review, which involves taking decisions on the network control at certain time instants, thus reducing the overall control overhead as opposed to schemes which require computations in every slot, such as in \cite{stai2016performance}. Discrete review schemes have been used in queueing networks \cite{Marg00}; however, the implementation is centralized and they do not consider delay deadlines. The use of discrete review separates our policy from works such as \cite{shroff13} or \cite{ashok2016distributed} which involve decision making in each slot. The policies in \cite{shroff13}, \cite{srikant} are throughput optimal but do not provide other QoS.
	\item Iterative gradient ascent is used to  solve the optimization problem in a distributed manner, similar to what is done in \cite{ashokFluid1}. This can be implemented easily in a cyclic manner, with message passing between the nodes after each step. The gradient calculation requires only local information, and the projection step requires knowledge of links with which a particular link interferes.
	\item The algorithm works based on queue length information, which acts as a proxy for delay. Thus it differs from \cite{ashokFluid1} which uses delay information to obtain delay guarantees, and thus has a different function being maximized. In addition, this algorithm is analysed extensively theoretically and is shown to be throughput optimal. Simultaneously, it also has provisions for mean delay QoS.
\end{itemize} 
The rest of this paper is organized as follows. In Section II, we describe the system model,
 and formulate an optimization problem to address our requirements. In Section III, we describe the algorithm and its distributed implementation in detail. In Section IV, we obtain the fluid limit of the system under our algorithm, and show its throughput optimality. In Section V we detail the simulation results, followed by conclusions in Section VI.
\section{System Model and Problem Formulation}
 We consider a multihop wireless network (Fig. \ref{fig1}). The network is a directed graph $G=(V,E)$ with $V=\{1,2,..,N\}$ being the set of vertices and $E$, the set of links on $V$. The system evolves in discrete time denoted by $t\in \{0,1,2,...\}$. We have directional links, with  link $(i,j)$ from node $i$ to node $j$ having a time varying channel gain $\gamma_{ij}(t)$ at time $t$. At each node $i$, $A_i^f(t)$ denotes the cumulative process of exogenous arrival of  packets destined to node $f$, upto time $t$. The packets arrive as an  i.i.d sequence across slots,with mean arrival rate  $\lambda_i^f$. All traffic in the network with the same destination $f$ is called  \emph{flow} $f$; the set of all flows is denoted by $F$. Each flow has a fixed route to follow to its destination. At each node there are queues, with $Q_i^f(t)$ denoting the queue length at node $i$ corresponding to flow $f \in F$ at time $t$. The queues evolve as
\begin{align}
Q_i^f(t)= Q_i^f(0)+A_i^f(t)+\sum_{k\neq i}S_{ki}^f(t)-\sum_{j \neq i}S_{ij}^f(t), \label{actualQueue}
\end{align}
where $S_{ij}^f(t)$ denotes the cumulative number of packets of flow $f$ that are transmitted from node $i$ to node $j$ till time $t$. The vector of queues at time $t$ is denoted by $Q(t)$. Similarly we have the arrival vector $A(t)$ and the service vector $S(t)$. We will assume that the links are sorted into $M$ \emph{interference sets}\ $I_1, I_2, \dotsc, I_M$. At any time, only one link from an interference set can be active. A link may belong to multiple interference sets. In this work we will assume that any two links which share a common node will fall in the same interference set.  We assume that each node transmits at unit power. Then, the rate of transmission between node $i$ and node $j$ can be written as $\mu_{ij}(I(t),\gamma(t))$ where $\mu$ is some achievable rate function, $\gamma (t)=\{\gamma_{ij}(t)\}_{i,j}$ and $I(t)$ is the schedule at time $t$.\\ 
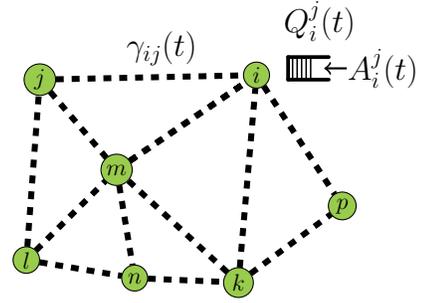
\begin{figure}
	\centering
	\setlength{\unitlength}{1cm}
	\thicklines
	\begin{tikzpicture}[scale=0.6, transform shape]		
	\node[draw,shape=circle, fill={rgb:orange,1;yellow,0;pink,2;green,2}, scale=0.6, transform shape] (v1) at (4.7,0.5) {\Huge $k$};
	\node[draw,shape=circle, fill={rgb:orange,1;yellow,0;pink,2;green,2}, scale=0.6, transform shape] (v2) at (2.4,0.6) {\Huge$n$};
	\node[draw,shape=circle, fill={rgb:orange,1;yellow,0;pink,2;green,2}, scale=0.6, transform shape] (v3) at (2.0,3.0) {\Huge$m$};
	\node[draw,shape=circle, fill={rgb:orange,1;yellow,0;pink,2;green,2}, scale=0.6, transform shape] (v4) at (0.3,5) {\Huge$j$};
	\node[draw,shape=circle, fill={rgb:orange,1;yellow,0;pink,2;green,2}, scale=0.6, transform shape] (v5) at (5.1,5.1) {\Huge$i$};
	\node[draw,shape=circle, fill={rgb:orange,1;yellow,0;pink,2;green,2}, scale=0.6, transform shape] (v6) at (7,2.2) {\Huge$p$};
	\node[draw,shape=circle, fill={rgb:orange,1;yellow,0;pink,2;green,2}, scale=0.6, transform shape] (v10) at (0,1.0) {\Huge$l$};
	\node (v7) at (7.9,5.25) {\huge $A_i^j(t)$};
	\node (v8) at (6.5,6.3) {\huge $Q_i^j(t)$};
	\node (v9) at (3,5.7) {\huge $\gamma_{ij}(t)$};			
	\draw[line width=0.8mm, dashed] (v2) -- (v1)
	(v4) -- (v5)			
	(v3) -- (v5)
	(v3) -- (v2)
	(v1) -- (v5)
	(v1) -- (v3)
	(v4) -- (v3)
	(v3) -- (v5)
	(v5) -- (v6)
	(v10) -- (v2)
	(v10) -- (v4)
	(v10) -- (v3)
	(v1) -- (v6);
	\draw[line width=0.5mm, line cap=round](5.8,5)--(6.7,5);
	\draw[line width=0.5mm, line cap=round](5.8,5.5)--(6.7,5.5);
	\draw[line width=0.5mm, line cap=round](5.8,5)--(5.8,5.5);
	\draw[line width=0.2mm](5.9,5)--(5.9,5.5);
	\draw[line width=0.2mm](6.0,5)--(6.0,5.5);
	\draw[line width=0.2mm](6.1,5)--(6.1,5.5);
	\draw[line width=0.2mm](6.2,5)--(6.2,5.5);
	\draw[line width=0.2mm](6.3,5)--(6.3,5.5);
	\draw[thick,->] (7.1,5.25) -- (6.6,5.25);
	\end{tikzpicture}
	\caption{A simplified depiction of a Wireless Multihop Network}
	\label{fig1}
\end{figure}
\indent We want to develop scheduling policies such that the different flows obtain their  end-to-end mean delay deadline guarantees. Our network control policy, Queue Weighted Discrete Review (QWDR), is as follows. We have a sequence of \emph{review times} $0=T_0<T_1<T_2....$, chosen as
\begin{align}
T_{i+1}=T_i+\max(1,\log(1+k_0||Q(T_i)||)), \label{ReviewTime}
\end{align}
 where the $\hat{T}_i:=T_{i+1}-T_i$ are called \emph{review periods}  and $||Q(t)||=\sum_{i,f}Q_i^f(t)$. Define $Q_{ij}^f=\max(Q_i^f-Q_j^f,0),  Q^f(t)=\sum_i Q_{i}^f(t)$.
 At each $T_i$, we solve the optimization problem,
 \begin{gather}
 \max\sum_{i,j,f}w(Q^f(T_i),\overline{Q}^f) Q_{ij}^f(T_i)\zeta_{ij}^f\mu_{ij}, \label{optFun}\\
 s.t\ 0\leq\zeta_{ij}:=\sum_{f\in F}\zeta_{ij}^f\leq 1\ \forall ij, \label{optCon1}\\
 0\leq\zeta_{ij}+\zeta_{kl}\leq 1,\ \forall (i,j),(k,l)\in I_m, \forall m, \label{optCon2}
 \end{gather}
 assuming $Q_{ij}^f>0$ for at least one link flow pair $(i,j),f$. If all $Q_{ij}^f$ are zero, we define the solution to be $\zeta_{ij}^f=0$ for all $i,j,f$. The first constraint corresponds to the fact that flows cannot simultaneously be scheduled on a link, and the second constraint corresponds to interference constraints. In (\ref{optFun}), we optimize the sum of rates  weighted by the function $w$ as well as the queue lengths. More weight may be given to  flows with larger backlogs, while the $w$ function captures the delay requirement of the flow. These are chosen such that flows requiring a lower mean delay would have a higher weight compared to flows needing a higher mean delay. Also, flows whose mean delay requirements are not met should get priority over flows whose requirements have been met. The weights $w$ therefore are functions of the state, and $\overline{Q}^f$ denotes a desired value for the queue length of flow $f$. We use the function
 \begin{align}
 w(x, \overline x)=1+\frac{a_1}{1+\exp (-a_2(x-\overline x))}. \label{logiF}
 \end{align}
 \noindent Thus $w$ is close to $1+a_1$ when $x$ is larger than $\overline x$, and reduces to $1$ as $x$ reduces. Thus, delays which are above certain thresholds obtain higher weights in the optimization function.  We seek to regulate the queue lengths using  $w$  with a careful selection of  $\overline{Q}^f$, and thereby control the delays. For any flow, the $\overline{Q}^f$ are chosen in the following manner. If the required end-to-end mean delay of the flow with arrival rate $\lambda$ is $\overline D$, we choose $\overline{Q}^f=\lambda\overline D$. In some sense, we are taking the queue length equivalent to the required delay using Little's Law and using it as a threshold that determines the scheduling process.\\ 
\indent The network control variables $\zeta_{ij}^f$ correspond to the fraction of time in one review period in which link $(i,j)$ will be transmitting flow $f$. In a review period, we will assume that the channel gain is fixed (slow-fading), but drawn as an \emph{i.i.d} sequence from a bounded distribution $\pi$. Each node transmits at unit power. The rate over link $(i,j)$ is  $\mu_{ij}=\log(1+\frac{\gamma_{ij}}{\sigma^2})$. Let $C_m(t)$ be the number of time slots till time $t$ in which the channel was in state $m$.  Let $C_{ijf}^{mI}(t)$ be the number of slots till time $t$, in which channel state was $m$, the schedule was $I$ and flow $f$ was scheduled over $(i,j)$. Clearly, for any $i$, $j$, $f$,
\begin{align}
\sum_{I}C_{ijf}^{mI}(t)=C_m(t). \label{CIdentity}
\end{align}
 \section{Gradient Ascent and Distributed Implementation}
The optimization problem is separable into link-flow elements, with each link-flow element being a unique $(i,j,f)$ in the network. Let $\mathcal{K}$ be the set of all link-flow elements. Any $k\in\mathcal{K}$ corresponds one-to-one with a link-flow element $(i,j,f)$; we  call this mapping from all $(i,j,f)$ to $\mathcal{K}$ as $\phi$. Consider the  optimization problem
 \begin{align}
 \max_{\zeta\in\mathcal{S}}\sum_{k\in\mathcal{K}} G_k(\zeta)
 \end{align}
 with $G_k(\zeta)=w(Q^f,\overline{Q}^f) Q_{ij}^f\zeta_{ij}^f\mu_{ij}$, and $\zeta=\{\zeta_{ij}^f\}_{i,j,f}$ where $k=\phi(i,j,f)$ and $\mathcal{S}$ is the set of $\zeta$ that satisfy  constraints (\ref{optCon1}) and (\ref{optCon2}); however, we remove the assumption that the $\zeta$ variables are positive. This is equivalent to the optimization problem (\ref{optFun}) . This is a linear optimization problem with linear constraints.  One can then define a sequential iteration
 \begin{align}
 \zeta^{j+1}=\Pi_{\mathcal{S}}[\zeta^j+\alpha\nabla G_{k_j}(\zeta^j)] \label{eqnUpdateIter}
 \end{align}
 with $\zeta^0$ being an arbitrary initial point,  $k_j=j$ modulo $|\mathcal{K}|+1$, and $\Pi_{\mathcal{S}}$ denoting projection into the set $\mathcal{S}$. This is cyclic Incremental Gradient Ascent. Let $\max_{i,j,f} w(Q^f,\overline{Q}^f) Q_{ij}^f \mu_{ij}= c_1$. 
 From Proposition 3.2 of \cite{Ber10}, the following holds.
 \begin{lem}
 	The iterates $\zeta^j$ given by equation (\ref{eqnUpdateIter}) satisfy
 	\begin{align*}
 	\lim_{j\to\infty}\sup \sum_{k\in\mathcal{K}} G_k(\zeta^j) \geq G^*-c,
 	\end{align*}
 	where  	$G^*= \max_{s\in\mathcal{S}}\sum_{k\in\mathcal{K}} G_k(s)$, $c=\frac{\alpha(4+|\mathcal{K}|^{-1}) |\mathcal{K}|^2c_1^2}{2}$.
 \end{lem}
Thus, given the optimization problem to be solved at a particular time, we can use the gradient ascent method (\ref{eqnUpdateIter}) to arrive at an optimal point in a distributed sequential fashion. First, obtain $\zeta^j+\alpha\nabla G_{k_j}(\zeta^j)$, and then  project onto $\mathcal{S}$. Since $\nabla G_{k}=w(Q^f,\overline{Q}^f) Q_{ij}^f\mu_{ij}$ where $k=\phi(i,j,f)$, the first step is clear. The projection step is described below.
  \subsection{Distributed Projection} 
   Two links that share a node are assumed to interfere with each other. Therefore, an update of the optimization variables at a $k=\phi(i,j,f)$ will affect only those links  which have either $i$ or $j$ as end points. The set $\mathcal{S}$ is defined by the intersection of half-spaces $\{\mathcal{H}_i\}_{i=1}^M$, where each  $\mathcal{H}_i$ is characterized by an equation $\langle \zeta,\nu^i\rangle\leq \beta_i$,
   with $\nu^i$ being the unit normal vector. Due to the nature of our constraints, $\nu^i$ is non-negative.\\   
   \indent Each  half-space corresponds to one interference constraint. During an update step, a point may break at most two interference constraints. This is because each link has two sets of constraints, one for each end. If one constraint is broken, one step of projection will suffice. If we break both constraints, we can iteratively project it, first to one hyperplane, then the next and so on repeatedly.   
   It can be shown \cite[Theorem 13.7]{Neumann} that this scheme  converges to the projection  onto  the intersection of the hyperplanes.  We will now obtain the analytical expressions for projecting a point onto a hyperplane. Let the hyperplane $\mathcal{H}$   be defined by $\langle \zeta,\nu\rangle\leq \beta$.
   Let the point $s$ lie outside $\mathcal{S}$, i.e.,
   \begin{align*}
   \beta^*\triangleq\langle s,\nu \rangle >\beta.
   \end{align*}
   Define $r=s-(\beta^*-\beta)\nu$. It is easy to see that $r$ satisfies $\langle r, \nu \rangle =\beta$. Since 	$s-r=(\beta^*-\beta)\nu$, and  $\nu$ is normal to the plane boundary of $\mathcal{H}$, it follows that $r$ is the perpendicular projection of $s$ onto  $\mathcal{H}$. It can also be shown  that $r$ satisfies all the other hyperplane constraints that $s$ does.  
   
    Now we will describe the complete algorithm.
   
\subsection{Algorithm Description}
 At each $T_i$, the problem (\ref{optFun})-(\ref{optCon2}) is solved in a distributed fashion. The nodes calculate $\zeta_{ij}^f$  for all $i$, $j$ and $f$, and  use these till the end of the review cycle.  We will now describe how  $\zeta_{ij}^f$  are calculated at each node.\\
\indent We assume that there is some convenient ordering of the link-flow elements, and computation proceeds in that order. Let this order be $k_1,k_2,...,k_{|\mathcal{K}|}$, and assume that the elements of the vector $s$ are also arranged in the same order. At link-flow element $k_1$, we  update the first component of $s$ as
\begin{align}
s(1)=s(1)+\alpha \nabla G_{k_1}(s).\label{nodeWiseUpdates}
\end{align}
Here $\nabla G_{k_1}(s)=w(Q^f,\overline{Q}^f) Q_{ij}^f\mu_{ij}$, where $\phi(i,j,f)=k_1$. The node then calculates the inner products
\begin{align*}
\beta_1^*\triangleq\langle s,\nu^1\rangle, \beta_2^*\triangleq\langle s,\nu^2\rangle
\end{align*} 
where $\langle s,\nu^1\rangle \leq \beta_1, \langle s,\nu^2\rangle \leq \beta_2$, correspond to the two interference constraints that the update step may break. These correspond to  constraints at the two nodes at which link $(i,j)$ is incident. If exactly one constraint, say $\beta_1$, gets broken, we can project the point back to the constraint set by calculating $\beta_{ex}=\frac{\beta_1^*-\beta_1}{N_1}$ where $N_1$ is the number of links in that interference set. The element communicates $\beta_{ex}$ to all elements in its interference set. All these elements do the update
\begin{align*}
s(k)=s(k)-\beta_{ex}.
\end{align*} 
If both constraints are violated, the above projection step has to be repeatedly done, first for elements corresponding to constraint $\beta_1$, then for $\beta_2$, again for $\beta_1$ and so on.\\ 
\indent  After projection, the node passes $s(1)$ to the node corresponding to the next component of  $s$, and the process is repeated cyclically, i.e, we repeat step (\ref{nodeWiseUpdates}) with 1 replaced by 2, and then by 3 and so on,  across the nodes till a stopping time. At the stopping time, set any negative components of $s$ to zero. For each interference set $I$, we check its constraint $\langle s,\nu\rangle \leq \beta$. If the constraint is not met, do an appropriate scaling. This ensures compliance with the constraints.\\
\indent The complete algorithm is given below, as Algorithm \ref{DistriAlgo}, QWDR (Queue Weighted Discrete Review), which uses in turn, Algorithms \ref{NodeWise}, \ref{NodeProject} and \ref{ScheCreate}. The last algorithm  schedules flows on a link for a fraction of time equal to the corresponding $s(k)$.
\vspace{-5mm}
\begin{algorithm}[h]
	\caption{\textcolor{black}{Algorithm QWDR}}
	\label{DistriAlgo}
	\begin{algorithmic}[1]
		\State $T_O=0$, $T_N=0$, $t\geq 0$				
		\If{$t=T_N$} 			
		\State $T_{O}\gets t, T_{N}\gets t+\max(1,\log(1+k_0\sum_{i,f}Q_i^f(t)))$
		\State Get $s_{ij}^f(T_{O})$ using Algorithm \ref{NodeWise}	
		\State Create schedule using Algorithm \ref{ScheCreate}	
		\EndIf
		\For{all $i,j,f$}
		\If{$\hat S_{ij}^f(t)=1$} \ schedule flow $f$ on link $(i,j)$
		\EndIf
		\EndFor			
	\end{algorithmic}	
\end{algorithm}
\vspace*{-6mm}
\begin{algorithm}[h]
	\caption{\textcolor{black}{Algorithm at node level}}
	\label{NodeWise}
	\begin{algorithmic}[1]
		\State Stopping time $T_s$, $t^{'}=0$, $s_{ij}^f(T_{O})=0$ for all $i,j,f$
		\While{$t^{'}<T_s$}
		\State $k=t^{'}\Mod{|\mathcal{K}|}+1$, $(i,j,f)\ s.t.\ \phi(i,j,f)=k$
		\State $s_{ij}^f(T_O)\gets s_{ij}^f(T_O)+\alpha w(Q_i^f,\overline{Q}_i^f) Q_{ij}^f\mu_{ij}$
		\State Project $s_{ij}^f(T_O)\gets \Pi_{\mathcal{S}}({s_{ij}^f(T_O)})$ using Algorithm \ref{NodeProject}		
		\State $t^{'} \gets t^{'}+1$
		\EndWhile
		\State $s_{ij}^f\gets \max (s_{ij}^f,0)$
		\State If $|s|:=\sum_{j,f} s_{ij}^f+\sum_{j,f} s_{ji}^f>1$, $s_{ij}^f\gets \frac{s_{ij}^f}{|s|}$	
	\end{algorithmic}	
\end{algorithm}
\vspace*{-6mm}
\begin{algorithm}[h]
	\caption{\textcolor{black}{Algorithm for Projection}}
	\label{NodeProject}
	\begin{algorithmic}[1]
		\State Link interference constraints $\langle s,\boldsymbol{\nu}^1\rangle\leq \beta_1, \langle s,\boldsymbol{\nu}^2\rangle\leq \beta_2$
		\State Calculate  $\beta_1^*\triangleq\langle s,\boldsymbol{\nu}^1\rangle, \beta_2^*\triangleq\langle s,\boldsymbol{\nu}^2\rangle$
		\If{$\beta_i^*>\beta_i$  and $\beta_j^*<\beta_j$}
		\State  $\beta_{ex}=\frac{\beta_i^*-\beta_i}{N_i+1}$, $N_i=$ number of interferers.
		\State For all interferers and current link, update $\beta_{ij}^f-\beta_{ex}$.
		\EndIf
		\If {$\beta_1^*>\beta_1$ and $\beta_2^*>\beta_2$}
		\State Repeat 4 to 6 $N\_rep$ times, sequentially for i, j
		\EndIf
	\end{algorithmic}	
\end{algorithm}
\begin{algorithm}[h]
	\caption{Algorithm for Schedule Creation}
	\label{ScheCreate}
	\begin{algorithmic}[1]
		\State Initialize $\hat S_{ij}^f(t)=0\ \forall i,j,f,t$
		\For{$k \leq N$}
		\State Obtain $\hat S_{ij}^f(t)$ for $i\leq k-1$
		\State Obtain $s_{kj}^f(T_{O})$ for all $j,f$
		\State Set of links that interfere with node $k =: N_k$
		\For{$j \in N_k, f\in F, t\in[T_{O},T_{N}]$}
		\If {$(\sum_{i\leq k-1}\hat S_{ij}^f(t)\sum_{i\in N_j}\hat S_{ji}^f(t)=0$ and $\sum_{t^o=T_{O}}^{t}\hat S_{kj}^f(t^o)<s_{kj}^f(T_{N}-T_{O})$} $\hat S_{kj}^f(t)=1$
		\EndIf
		\EndFor		
		\EndFor
	\end{algorithmic}					
\end{algorithm}
\vspace{-6mm}
 \section{Fluid Limit}
Define the system state to be $Y(t)=(Q(t),\tilde Q(t),\tilde S(t))$, where the process $\tilde Q(t)=Q(T)$ with $T=\sup\{s\leq t:s=T_i \text{\ for some }i\}$, representing the queue values at the last review instant, and $\tilde S_{ij}^f(t)=S_{ij}^f(t)-S_{ij}^f(T)$ representing the cumulative allocation vector from the last review instant to the current time. From the queue evolution (\ref{actualQueue}) and the allocation, it is clear that the system $Y(t)$ evolves as a discrete time countable Markov chain, since at any time $t$ the next state may be computed by solving the optimization (\ref{optFun}) with $Q$ replaced by $\tilde Q$, and using the cumulative allocation process $\tilde S$ to determine how allocation must be done in the next slot to satisfy the solution of (\ref{optFun}). The associated norm is $||Y(t)||=\sum_{i,f}(Q_i^f+\tilde{Q}_i^f)+\sum_{i,j,f}\tilde{S}_{ij}^f$. Positive recurrence of this Markov chain would imply stability. We will show positive recurrence of this Markov chain via its fluid limit.\\ 
 \indent Consider a real valued process $h^n(t)$ evolving over (discrete) time $t$, with $n$ being its initial state. Consider a sequence of these processes as $n\to\infty$. Define the corresponding scaled (continuous time) process, 
  \begin{align*}
  h(n,t)=\frac{h^{n}(n \left\lfloor t \rfloor\right)}{n}.
  \end{align*}
Define the scaled processes $A_i^f(n,t)$, $Q_i^f(n,t)$, $\tilde Q_i^f(n,t)$, $\tilde S_{ij}^f(n,t)$ and $S_{ij}^f(n,t)$ as above.  For a scaled process $h(n,t)$, denote $h^n=\{h(n,t),t\geq 0\}$. For the vector processes  $A(t)$, $Q(t)$, $\tilde{Q}(t)$, $S(t)$ and $\tilde S(t)$, we define the corresponding scaled vector processes. The term fluid limit denotes the limits obtained as we scale $n\to\infty$ for these processes. Consider $\mathcal{X}=\{Q,\tilde{Q},A,S,\tilde S,C\}$. The process $Y$ is a projection of $\mathcal{X}$.\\
\indent We assume that the rates satisfy $\mu_{ij}(t)\leq\mu_{max}$. This will happen since the channel gains are assumed bounded  and transmit power is fixed. Consider the scaled process $\mathcal{X}^n=\{Q^n,\tilde{Q}^n,A^n,S^n,\tilde S^n,C^n\}$.
 We use the following definition.
 \newtheorem{definition}{Definition}
 \begin{definition}
 	A sequence of functions $g_n$ is said to converge uniformly on compact sets (u.o.c) if $g_n\to g$ uniformly on every compact subset of the domain. 
 \end{definition}
We obtain the following result for the components of $\mathcal{X}^n$. 
\begin{thm}
	For almost every sample path $\omega$, there exists a subsequence $n_k(\omega)\to\infty$ such that,
	\begin{gather}
	A_i^f(n_k,t) \to a_i^f(t),\ \ \ S_{ij}^f(n_k,t) \to s_{ij}^f(t),\label{convS}\\
	C_m(n_k,t) \to c_m(t),\ \ \ C_{ijf}^{mI}(n_k,t) \to c_{ijf}^{mI}(t), \label{convC}\\
		Q_i^f(n_k,t) \to q_i^f(t),\ \ \ \tilde{Q}_i^f(n_k,t) \to q_i^f(t), \label{convQtilde}\\	\tilde{S}_i^f(n_k,t) \to 0, \label{convStilde}
	\end{gather}
	uniformly on compact sets, for all $i$, $j$ and $f$.
	The limiting functions are also Lipschitz continuous, and hence almost everywhere differentiable. The points $t$ at which it is differentiable are called regular points.	
	In addition, the limiting functions satisfy the following properties.
	\begin{align}
	a_i^f(t) = \lambda_i^f t,\ \ \ \  c_m(t) = \pi_m t, \label{limitChan}
	\end{align}	
	\begin{align}
	q_i^f(t) &= q_i^f(0) +a_i^f(t)+\sum_{k}s_{ki}^f(t)-\sum_{j}s_{ij}^f(t), \label{fluQu}
	\end{align}
	\begin{align}
	\dot{q}_i^f(t) = \lambda_i^f+\sum_{k}\dot{s}_{ki}^f(t)-\sum_{j}\dot{s}_{ij}^f(t), \label{fluQuDeriv}
	\end{align}
	\begin{align}
	\sum_I c_{ijf}^{mI}(t) = c_m(t),\ \ \  ||q(0)|| \leq 1, \label{InitCondtn}
	\end{align}
	\begin{align}
	s_{ij}^f(t) = \int_0^t \dot{s}_{ij}^f(\tau)d\tau, \label{fluAlloc}
	\end{align}
	where  $\dot{s}(t)$ satisfies
	\begin{align}
\sum_{i,j,f} w(q^f(t))q_{ij}^f(t)\dot{s}_{ij}^f(t) = \max_{\bar\mu} \sum_{i,j,f}w(q^f(t))q_{ij}^f(t)\bar{\mu}_{ij}, \label{allocDeriv}
	\end{align}
	where the dot indicates derivative, at regular t and $\bar{\mu}=\sum_m \pi_m\mu(m,I)$.
\end{thm}
\begin{proof}
The Strong Law of Large Numbers implies
\begin{align*}
A_i^f(n_k,t) \to a_i^f(t),
\end{align*}
for any subsequence $n_k\to\infty$, with $a_i^f(t)=\lambda_i^f t$. Thus we obtain the first parts of (\ref{convS}) and (\ref{limitChan}).
Since the rates are bounded, it follows that $S_{ij}^f(t)\leq\mu_{max}t$. Therefore, for $0\leq t_1\leq t_2$, we have
\begin{align*}
S_{ij}^f(n t_2)-S_{ij}^f(n t_1)\leq n(t_2-t_1)\mu_{max}.
\end{align*}
Thus, the family of functions $\{\frac{1}{n} S_{ij}^f(nt)\}$ is uniformly bounded and equicontinuous. By the Arzela-Ascoli theorem \cite{babyRudin}, we can see that for any sequence $S_{ij}^f(n,t)$ with   $n\to\infty$, there exists a subsequence $n_k\to\infty$ along which
\begin{align*}
S_{ij}^f(n_k,t)\to s_{ij}^f(t),
\end{align*}
as $n_k\to\infty$ wp 1, u.o.c. This implies the second part of (\ref{convS}). The resultant $s_{ij}^f$ is Lipschitz, being the result of uniform convergence of a sequence of Lipschitz functions. The first part of (\ref{convC}), and second part of (\ref{limitChan}) follow from the Strong Law of Large Numbers applied to the channel process. From equation (\ref{actualQueue}), we can see that the terms on the right hand side converge under this scaling. Consequently, 
\begin{align*}
{Q}_{i}^f(n_k,t)\to q_{i}^f(t),
\end{align*}
wp 1 u.o.c, as $n_k\to \infty$. Like $s_{ij}^f(t,\omega)$, both $a^f_i(t,\omega)$ and $q_i^f(t,\omega)$ will be Lipschitz.  Equation (\ref{fluQu}) follows by observing that the scaled queue process will satisfy the queueing equation (\ref{actualQueue}), and applying the appropriate limit in that equation.
\indent Since the fluid variables $q,a$ and $s$ are Lipschitz, they are differentiable almost everywhere. At the points where they are differentiable, we obtain (\ref{fluQuDeriv}) by differentiating (\ref{fluQu}). The first part of (\ref{InitCondtn}) follows from (\ref{CIdentity}).\\
\indent To see the second part of (\ref{convC}), observe that, by definition,
\begin{align}
\frac{1}{n}C_{ijf}^{mI}(nt_2)-C_{ijf}^{mI}(nt_1)\leq t_2-t_1,
\end{align}
for $t_2>t_1$. Applying Arzela-Ascoli theorem, we obtain the subsequence that satisfies the second part of (\ref{convC}).\\
\indent Since $s$ is almost everywhere differentiable, (\ref{fluAlloc}) follows.
 In obtaining the fluid limit of the allocation process $S$, we will not distinguish between the actual and the ideal allocation, since they converge to the same limit. Let the actual allocation be $\hat{S}_{ij}^f(t)$. The actual allocation differs from the ideal allocation due to round-off errors. At a time $nt$, let $m=\max\{i:T_i\leq nt\}$. Bounding possible errors in each review period we get,
\begin{align*}
|\hat S_{ij}^f(nt)-S_{ij}^f(nt)|\leq \mu_{max}\hat T_{m}+m\mu_{max}.
\end{align*}
The last term follows by summing up round-off errors in review periods upto $m$, and observing that in any review period $\hat T$, errors are of the form $\mu_{ij}|x-\lfloor x\rfloor|$, where $x=\zeta_{ij}^f\hat T$. Since $m\leq \frac{nt}{T}$, where $T=\min_{i<m}\{\hat{T}_i\}$, we get
 \begin{align*}
 |\hat S_{ij}^f(n,t)-S_{ij}^f(n,t)|\leq \mu_{max}\{\frac{\hat{T}_m}{n}+\frac{t}{T}\}.
 \end{align*}
 Since $\hat T_i$ are $\max(1,\log (1+k_0||Q||))$ and $\lim_{n\to\infty}||Q||=\infty$, we have $\lim_{n\to\infty}T=\infty$ and $\lim_{n\to\infty}\frac{\hat{T}_m}{n}=0$, and hence, the fluid limits of $\hat S$ and $S$ are equal.\\
 \indent To show (\ref{allocDeriv}), observe that
\begin{align*}
S_{ij}^f(t)=\sum_{m,I} C_{ijf}^{mI}(t)\mu_{ij}(I,m).
\end{align*}
Hence we have
\begin{align*}
S_{ij}^f(nt_2)-S_{ij}^f(nt_1)=\sum_{m,I} (C_{ijf}^{mI}(nt_2)-C_{ijf}^{mI}(nt_1))\mu_{ij}(I,m).
\end{align*}
 Multiplying LHS and RHS by $w(Q^f(n,t_1))Q_{ij}^f(n,t_1)$, summing over \emph{i}, \emph{j}, \emph{f}, and taking $n\to\infty$, the LHS becomes\\
\begin{align}
\sum_{i,j,f} w(q^f(t_1))q_{ij}^f(t_1)[s_{ij}^f(t_2)-s_{ij}^f(t_1)], \label{LHSeq}
\end{align}
where $q_{ij}^f(t)=\max(q_i^f(t)-q_j^f(t),0)$ and $q^f(t_1)=\lim_{n\to\infty}Q^f(n,t_1)=\sum_i q_i^f(t)$.
Since the allocation satisfies
\begin{align*}
\sum w(Q^f(n,t^{'}))Q_{ij}^f(n,t^{'})\mu_{ij}(I,m)\\
=\max_I \sum w(Q^f(n,t^{'}))Q_{ij}^f(n,t^{'})\mu_{ij}(I,m),
\end{align*} 
where $nt^{'}$ was the previous review point with $nt_1=nt^{'}+T$. Since $\frac{T}{n}\to 0$, we can write $q_i^f(t_1)$ as
\begin{align}
\lim_{n\to\infty}Q_i^f(n,t_1)=\lim_{n\to\infty}\frac{1}{n}Q_i^f(n(t^{'}+\frac{T}{n}))=q_i^f(t^{'}). \label{FluLiEqui}
\end{align}
The RHS can therefore be written as
\begin{align*}
\sum_{m,I,i,j,f}[c_{ijf}^{mI}(t_2)-c_{ijf}^{mI}(t_1)]\max_I\mu_{ij}(I,m)w(q^f)q_{ij}^f.
\end{align*}
Using (\ref{InitCondtn}) and (\ref{limitChan}), this becomes
\begin{align}
&\sum_{m,i,j,f} [c_m(t_2)-c_m(t_1)]\max_I\mu_{ij}(I,m)w(q^f)q_{ij}^f,\\
= &(t_2-t_1)\sum_{m,i,j,f}\pi_m\max_I\mu_{ij}(I,m)w(q^f)q_{ij}^f. \label{RHSeq}
\end{align}
Dividing (\ref{LHSeq}) and (\ref{RHSeq}) by $t_2-t_1$, equating, and taking $t_2\to t_1$,
\begin{align*}
&\sum_{i,j,f} w(q^f(t_1))q_{ij}^f(t_1)\dot{s}_{ij}^f(t_1)= \max_{\bar\mu} \sum_{i,j,f}w(q^f)q_{ij}^f\bar{\mu}_{ij},
\end{align*}
where $\bar{\mu}=\sum_m \pi_m\mu(m,I)$. Thus we obtain (\ref{allocDeriv}). The second part of (\ref{convQtilde})follows from (\ref{FluLiEqui}). To obtain (\ref{convStilde}), observe that
\begin{align*}
0\leq \tilde{S}_{ij}^f(n,t)\leq \mu_{max}\frac{\hat T}{n},
\end{align*}
with $\hat T$ being a review period. Taking $n\to\infty$, (\ref{convStilde}) follows. Since $||Q(n,0)||=1$, the second part of (\ref{InitCondtn}) follows. 
\end{proof}
\vspace{-2mm}
Denote the vector of all $q_i^f(t)$ by $q(t)$. We will use the following result to establish the stability of the network.
\begin{thm}\label{StabLem}(Theorem 4 of \cite{andrews2004scheduling})
	Let $Y$ be a Markov Process with $||Y(.)||$ denoting its norm. If there exist $\alpha>0$ and a time $T>0$ such that for a scaled sequence of processes $\{Y^n,n=0,1,2,..\}$, we have
	\begin{align*}
	\lim_{n\to\infty}\sup \mathbb{E}[||Y(n,T)||]\leq 1-\alpha,
	\end{align*}
	then the process $Y$ is stable (positive recurrent).
\end{thm}

Using this result, we will establish stability of the network under our algorithm and show that it is throughput optimal. We first define the capacity region of the network. A \emph{schedule} $\textbf{s}$ is a mapping from $\mathcal{K}$, the set of all link-flow elements, to $[0,1]$.  Let the set of all feasible schedules be denoted by $\mathcal{S}$ and $\Gamma$ is its convex hull. We define the capacity region as follows.
\begin{definition}
	The capacity region $\Lambda$ is given by the set of all $\lambda\geq 0$ for which there exists an $\textbf{s}\in\Gamma$ such that
	\begin{align}
	\lambda_i^f \leq \sum_m\pi_m\sum_j \textbf{s}(k) \mu_{ij}^m-\sum_m\pi_m\sum_r \textbf{s}(k_r) \mu_{ri}^m, \label{CapReEq}
	\end{align}
	where $\phi(i,j,f)=k,\ \phi(r,i,f)=k_r$, $\pi_m$ is the stationary probability that the  channels are in state $m$, and $\mu_{ij}^m$ is the rate across link $(i,j)$ when the channels are in state $m$.
	\vspace{-2mm}
\end{definition}
Now we establish the throughput optimality of our policy.
	\vspace{-2mm}
\begin{thm}
	The policy QWDR stabilizes the process $\{Q(t),t\geq 0\}$ for all arrivals in the interior of $\Lambda$.
\end{thm}
\vspace{-4mm}
\begin{proof}
	To prove this, we will first pick a suitable Lyapunov function, whose drift will be shown to be negative. \\	
	\indent Pick an arrival rate matrix $\lambda=\{\lambda_i^f\}\in int(\Lambda)$. This implies that there are rates and $\epsilon>0$ that satisfy 
	\begin{align}
	\lambda_i^f+\epsilon < \sum_n r_{in}^f-\sum_m r_{mi}^f,\label{lambInt}
	\end{align}
	for each $i,f$. These rates correspond to the terms in (\ref{CapReEq}).	
	\indent  Consider the Lyapunov function
	\begin{align*}
	L(q(t))=-\int_t^{\infty}\exp(\tau-t)\sum_{i,f} w(q^f)q_{i}^f\dot{q}_i^f d\tau,
	\end{align*}
	where the dot indicates the derivative. This is a continuous function of $q(t)$, with $L(0)=0$. We can write the derivative,
	\begin{align*}
	\dot L(q(t)) &=\sum_{i,f}  w(q^f)q_i^f\dot{q}_i^f,\\
	&= \sum_{i,f}   w(q^f)q_i^f(\lambda_i^f+\sum_m \dot{s}_{mi}^f(t)-\sum_n \dot{s}_{in}^f(t)),\\
	&< -\epsilon\sum_{i,f} w(q^f)q_i^f + \sum_{i,f}   w(q^f)q_i^f(\sum_n r_{in}^f,\\
	&\ \ \ \ -\sum_m r_{mi}^f+\sum_m \dot{s}_{mi}^f(t)-\sum_n \dot{s}_{in}^f(t)),
	\end{align*}
	where the inequality followed from (\ref{lambInt}). 
	Observing that
	\begin{align*}
	\sum_{i,f} w(q^f)q_i^f(\sum_n r_{in}^f-\sum_m r_{mi}^f)=\sum_{i,j,f} w(q^f)r_{ij}^f(q_i^f-q_j^f),
	\end{align*}
	and that a similar equation holds for $r$ replaced by $s$, 	it follows that if we show 
	\begin{align}
	\sum_{i,j,f} w(q^f) r_{ij}^f(q_i^f-q_j^f)\leq \sum_{i,j,f}  w(q^f)\dot{s}_{ij}^f(q_i^f-q_j^f), \label{cruxEqn}
	\end{align}
	it will imply $\dot L(q(t))<0$. 	We have
	\begin{align*}
	\sum_{i,j,f} w(q^f) r_{ij}^f(q_i^f-q_j^f)\leq  \sum_{i,j,f}  w(q^f)r_{ij}^fq_{ij}^f 
	\leq \sum_{i,j,f}  w(q^f)\dot{s}_{ij}^fq_{ij}^f,
	\end{align*}
	where the second inequality follows from (\ref{allocDeriv}).
	Now, if we show that $\dot{s}_{ij}^f=0$ whenever $q_{ij}^f= 0$, (\ref{cruxEqn}) will follow. To see this, assume that at some $t$, $\dot{s}_{ij}^f=\delta^{1}>0$ and $q_{ij}^f= 0$. This would mean that for large enough $n$, there is a time $s$ sufficiently close to $t$ such that, for  $\delta=\frac{\delta^{1}}{2}$,
	\begin{align*}
	S_{ij}^f(nt)-S_{ij}^f(ns)>n\delta(t-s).
	\end{align*}
	This implies that at a time $t_1\in(s,t)$ with  $Q_{i}^f(nt_1)-Q_{j}^f(nt_1)\leq 0$ the queue $Q_i^f$ was served. This would mean that the optimization (\ref{optFun}) resulted in a positive $\zeta_{ij}^f$. This cannot arise in a condition where all $Q_{ij}^f$ are zero, since in that state, by definition, all $\zeta_{ij}^f$ are set to zero. Hence there exists  $k,l,m$ such that  $Q_{kl}^m>0$. If  $\zeta_{ij}^f$  is added to $\zeta_{kl}^m$, the value of (\ref{optFun}) would only increase, thus contradicting its optimality. It follows that $\dot{s}_{ij}^f=0$ whenever $q_{ij}^f= 0$, and hence, (\ref{cruxEqn}) is true. \\
	\indent Thus, $\dot L(q(t))<-\epsilon \sum_{i,f}w(q^f)q_i^f$, and hence, $L(q(t))>0$ whenever $q(t)\neq 0$. Fix $\delta_1<0.5$. Then, there exists $T\leq T_1=\frac{L(q(0))}{\epsilon\delta_1}+\delta_1$ such that $\sum_{i,f}q_i^f\leq\delta_1$. To see this, assume otherwise, that $\sum_{i,f}q_i^f(t)>\delta_1$ for $t\in[0,T_1]$. Now,
	\begin{align*}
	L(q(t))=L(q(0))+\int_0^{t}\dot{L}(q(\tau))d\tau.
	\end{align*}
	Since $q$ is Lipschitz, $\dot q$ will be bounded.  Using (\ref{InitCondtn}) and the fact that $q(t)$ will grow at most linearly in $t$, it can be shown that $L(q(0))$ is finite. Since $w(q^f)\geq 1$,
	\begin{align*}
	L(q(t))\leq L(q(0))-\epsilon\delta_1t,
	\end{align*}
	for $t\in[0,T_1]$, and by choosing $t=T_1$, we obtain $L(q(T_1))<0$, which is a contradiction. Hence, $\sum_{i,f}q_i^f(T)\leq \delta_1$. Since the fluid queue is a deterministic process following the trajectory defined by equations (\ref{convC})-(\ref{allocDeriv}), it follows that, almost surely,
\begin{align*}
\lim_{n\to\infty}\sup ||Y(n,T)||=2\sum_{i,j,f}q(T)\leq 2\delta_1<1.
\end{align*}
 From the definition of $Y$, we have that
\begin{align*}
||Y(n,T)||\leq 2[1+ \sum_{i,f}A_i^f(n,T)+T\sum_{i,j,f}\mu_{max}].
\end{align*}
Since $\mathbb{E}[\sum_{i,f}A_i^f(n,T)]= T(\sum_{i,f}\lambda_i^f)<\infty$, we can use the Dominated Convergence Theorem \cite{athreya2006measure} to see that Theorem \ref{StabLem} holds for $Y$ with $\alpha=1-2\delta_1$. The result follows.
\end{proof}
\section{Simulation Results}
For simulation we consider a fifteen node network with seven flows, with connectivity as depicted in Fig \ref{figsim}, over a unit area. We will be trying to provide mean delay QoS for three of these flows. The channel gains are Rayleigh distributed with parameters inversely proportional to the square of the distance between nodes, and the arrival distribution is Poisson. The flows are $F10: 7\to 9\to 10$, $F4: 7\to 8\to 2 \to 4$, $F11:1\to 2 \to 4\to 11$, $F13: 9\to 10\to 13$, $F12: 1\to 3\to 6\to 12$, $F15: 5\to 14\to 15$ and $F6: 5\to 3\to 6$. The constant $k_0=0.01$ in (\ref{ReviewTime}), and in the distributed optimization, the algorithm runs 15 cycles over the set of nodes with $\alpha=0.0001$ and the initial state is zero. The simulation runs for $10^5$ slots, with $a_1=0.2$ and $a_2=2$ in (\ref{logiF}). The arrival rates are 3.8 for $F6$, 3.74 for $F10$, and 2.5 for the others. These are chosen to take the queues to the edge of the stability region, where delays are larger, and the control of the algorithm in providing QoS will be more evident. The values are shown in Table \ref{table:Two}, with flows $F10$, $F11$ and $F6$ having mean delay requirements, which are translated to $\lambda \bar{Q}^f$ in the $w$ function. The delays  are rounded to the nearest integer.\\
\indent The first row represents delays of the flows when $w=1$ for all flows, i.e., no priority is given. In the other rows, the values in brackets are of the form (target delay, achieved delay). The flows seem to respond very well to the target, often coming much lower than what is desired, since the weights tend to push the queue lengths to below these threshold values. In all cases, the delays can be brought down to less than $50\%$ of their unweighed values. Another effect is that giving QoS to one flow does not adversely affect the delay of the other flows. In fact, it can substantially reduce the mean delays of the other flows as well. Since the algorithm uses backpressure values, this is not surprising, and the weight function can be thought of as fine tuning the delay behaviour of the network.
\vspace*{-6mm}
\begin{figure}
	\centering
	\begin{tikzpicture}[scale=0.8, transform shape]
	\node[draw,shape=circle, fill={rgb:orange,1;yellow,2;pink,5}, scale=0.6, transform shape] (v1) at (3.2,3.6) {$1$};
	\node[draw,shape=circle, fill={rgb:orange,1;yellow,2;pink,5}, scale=0.6, transform shape] (v2) at (2.4,2.0) {$2$};
	\node[draw,shape=circle, fill={rgb:orange,1;yellow,2;pink,5}, scale=0.6, transform shape] (v3) at (3.2,2) {$3$};
	\node[draw,shape=circle, fill={rgb:orange,1;yellow,2;pink,5}, scale=0.6, transform shape] (v4) at (2.4,1.2) {$4$};
	\node[draw,shape=circle, fill={rgb:orange,1;yellow,2;pink,5}, scale=0.6, transform shape] (v5) at (4.3,1.8) {$5$};
	\node[draw,shape=circle, fill={rgb:orange,1;yellow,2;pink,5}, scale=0.6, transform shape] (v6) at (3.44,1.0) {$6$};
	\node[draw,shape=circle, fill={rgb:orange,1;yellow,2;pink,5}, scale=0.6, transform shape] (v7) at (1.2,3.6) {$7$};
	\node[draw,shape=circle, fill={rgb:orange,1;yellow,2;pink,5}, scale=0.6, transform shape] (v8) at (1.6,2.8) {$8$};
	\node[draw,shape=circle, fill={rgb:orange,1;yellow,2;pink,5}, scale=0.6, transform shape] (v9) at (0.4,2.8) {$9$};
	\node[draw,shape=circle, fill={rgb:orange,1;yellow,2;pink,5}, scale=0.6, transform shape] (v10) at (0.04,1.4) {$10$};
	\node[draw,shape=circle, fill={rgb:orange,1;yellow,2;pink,5}, scale=0.6, transform shape] (v11) at (2.04,0.4) {$11$};
	\node[draw,shape=circle, fill={rgb:orange,1;yellow,2;pink,5}, scale=0.6, transform shape] (v12) at (2.8,0.6) {$12$};
	\node[draw,shape=circle, fill={rgb:orange,1;yellow,2;pink,5}, scale=0.6, transform shape] (v13) at (0.4,0.8) {$13$};
	\node[draw,shape=circle, fill={rgb:orange,1;yellow,2;pink,5}, scale=0.6, transform shape] (v14) at (4.3,0.8) {$14$};
	\node[draw,shape=circle, fill={rgb:orange,1;yellow,2;pink,5}, scale=0.6, transform shape] (v15) at (3.8,0.4) {$15$};	
	\draw[line width=0.2mm] (v1) -- (v2)
	(v1) -- (v5)
	(v3) -- (v5)
	(v4) -- (v2)
	(v11) -- (v4)
	(v11) -- (v12)
	(v10) -- (v13)
	(v9) -- (v10)
	(v7) -- (v9)
	(v7) -- (v8)
	(v8) -- (v2)
	(v3) -- (v6)
	(v6) -- (v12)
	(v5) -- (v14)
	(v5) -- (v6)
	(v9) -- (v4)
	(v3) -- (v4)
	(v14) -- (v15)
	(v12) -- (v15)
	(v13) -- (v11)
	(v3) -- (v1);	
	\end{tikzpicture}
	\caption{Sample Network}
	\label{figsim}
\end{figure}
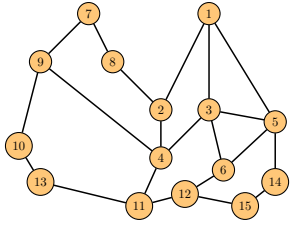
\begin{table}[h]
	\centering
	\caption{Simulation for example in Fig \ref{figsim}. Three Flows with mean delay requirements, network of fifteen nodes. Entries of the form (a,b) indicate delay target a, delay achieved b.}
	\label{table:Two}
	\begin{tabular}{|p{1.0cm}|p{0.3cm}|p{1.0cm}|p{0.3cm}|p{0.3cm}|p{0.3cm}|p{1cm}|}
		\hline
		\multicolumn{7}{|c|}{Mean Delay(slots) for each flow}\\		
		\hline 
		\multicolumn{1}{|c}{F10}\ & \multicolumn{1}{|c}{F4} & \multicolumn{1}{|c}{F11} & \multicolumn{1}{|c}{F13}\ & \multicolumn{1}{|c}{F12} & \multicolumn{1}{|c}{F15} & \multicolumn{1}{|c|}{F6} \\
		\hline
		318 & 68 & 499 & 233 & 642 & 25 & 111 \\
		\hline
		(200,188) & 61 & (350,304) & 163 & 403 & 23 & (70,70) \\
		\hline
		(150,96) & 60 & (300,265) & 85 & 362 & 22 & (60,65) \\
		\hline
		(150,67) & 56 & (150,148) & 61 & 235 & 22 & (45,55) \\
		\hline
		(200,136) & 56 & (130,134) & 119 & 220 & 22 & (50,55) \\
		\hline
	\end{tabular}
\end{table}
\vspace{-5mm}
\section{Conclusion}
We have developed a throughput optimal distributed algorithm to provide end-to-end mean delay requirements for flows in a wireless multihop network. The algorithm uses discrete review to solve an optimization problem at review instants, and uses a control policy based on the solution of an optimization problem. The optimization problem incorporates mean delay requirements as weights to the control variables. These weights vary dynamically, depending on the current state of the system. Iterative gradient ascent and distributed iterative projection are used to compute the optimal point in a distributed manner. We also demonstrate throughput optimality of the algorithm by a theoretical analysis. By means of simulations we demonstrate the effectiveness of the algorithm. Surprisingly, the algorithm  not only reduces the mean delays of the targeted flows, it reduces mean delays of other flows as well.
\bibliography{survey}

\begin{thebibliography}{10}
\providecommand{\url}[1]{#1}
\csname url@samestyle\endcsname
\providecommand{\newblock}{\relax}
\providecommand{\bibinfo}[2]{#2}
\providecommand{\BIBentrySTDinterwordspacing}{\spaceskip=0pt\relax}
\providecommand{\BIBentryALTinterwordstretchfactor}{4}
\providecommand{\BIBentryALTinterwordspacing}{\spaceskip=\fontdimen2\font plus
\BIBentryALTinterwordstretchfactor\fontdimen3\font minus
  \fontdimen4\font\relax}
\providecommand{\BIBforeignlanguage}[2]{{%
\expandafter\ifx\csname l@#1\endcsname\relax
\typeout{** WARNING: IEEEtran.bst: No hyphenation pattern has been}%
\typeout{** loaded for the language `#1'. Using the pattern for}%
\typeout{** the default language instead.}%
\else
\language=\csname l@#1\endcsname
\fi
#2}}
\providecommand{\BIBdecl}{\relax}
\BIBdecl

\bibitem{markopoulou2002assessment}
A.~P. Markopoulou, F.~A. Tobagi, and M.~J. Karam, ``Assessment of voip quality
  over internet backbones,'' in \emph{INFOCOM 2002. Twenty-First Annual Joint
  Conference of the IEEE Computer and Communications Societies. Proceedings.
  IEEE}, vol.~1.\hskip 1em plus 0.5em minus 0.4em\relax IEEE, 2002, pp.
  150--159.

\bibitem{shah2016remote}
T.~Shah, A.~Yavari, K.~Mitra, S.~Saguna, P.~P. Jayaraman, F.~A. Rabhi, and
  R.~Ranjan, ``Remote health care cyber-physical system: quality of service
  (qos) challenges and opportunities.'' \emph{IET Cyper-Phys. Syst.: Theory \&
  Appl.}, vol.~1, no.~1, pp. 40--48, 2016.

\bibitem{atzori2010internet}
L.~Atzori, A.~Iera, and G.~Morabito, ``The internet of things: A survey,''
  \emph{Computer networks}, vol.~54, no.~15, pp. 2787--2805, 2010.

\bibitem{awan2014modelling}
I.~Awan, M.~Younas, and W.~Naveed, ``Modelling qos in iot applications,'' in
  \emph{Network-Based Information Systems (NBiS), 2014 17th International
  Conference on}.\hskip 1em plus 0.5em minus 0.4em\relax IEEE, 2014, pp.
  99--105.

\bibitem{she2016energy}
C.~She and C.~Yang, ``Energy efficiency and delay in wireless systems: Is their
  relation always a tradeoff?'' \emph{IEEE Transactions on Wireless
  Communications}, vol.~15, no.~11, pp. 7215--7228, 2016.

\bibitem{LauSurvey}
Y.~Cui, V.~K. Lau, R.~Wang, H.~Huang, and S.~Zhang, ``A survey on delay-aware
  resource control for wireless systems—large deviation theory, stochastic
  lyapunov drift, and distributed stochastic learning,'' \emph{IEEE
  Transactions on Information Theory}, vol.~58, no.~3, pp. 1677--1701, 2012.

\bibitem{GeorgBook}
L.~Georgiadis, M.~J. Neely, and L.~Tassiulas, ``Resource allocation and
  cross-layer control in wireless networks,'' \emph{Foundations and
  Trends{\textregistered} in Networking}, vol.~1, no.~1, pp. 1--144, 2006.

\bibitem{Cui}
Y.~Cui, E.~M. Yeh, and R.~Liu, ``Enhancing the delay performance of dynamic
  backpressure algorithms,'' \emph{IEEE/ACM Transactions on Networking},
  vol.~24, no.~2, pp. 954--967, 2016.

\bibitem{sharma2007opportunistic}
V.~Sharma, D.~Prasad, and E.~Altman, ``Opportunistic scheduling of wireless
  links,'' in \emph{Managing Traffic Performance in Converged Networks}.\hskip
  1em plus 0.5em minus 0.4em\relax Springer, 2007, pp. 1120--1134.

\bibitem{Kumar}
R.~Singh and P.~R. Kumar, ``Throughput optimal decentralized scheduling of
  multi-hop networks with end-to-end deadline constraints: Unreliable links,''
  \emph{arXiv preprint arXiv:1606.01608}, 2016.

\bibitem{SatWiOPt}
V.~S. Kumar, L.~Kumar, and V.~Sharma, ``Energy efficient low complexity joint
  scheduling and routing for wireless networks,'' in \emph{Modeling and
  Optimization in Mobile, Ad Hoc, and Wireless Networks (WiOpt), 2015 13th
  International Symposium on}.\hskip 1em plus 0.5em minus 0.4em\relax IEEE,
  2015, pp. 8--15.

\bibitem{SatVS}
S.~V. Kumar and V.~Sharma, ``Joint routing, scheduling and power control
  providing hard deadline in wireless multihop networks,'' in \emph{2017
  Information Theory and Applications Workshop (ITA)}, San Diego, Feb. 12-17,
  2017.

\bibitem{Meyn08}
S.~Meyn, \emph{Control techniques for complex networks}.\hskip 1em plus 0.5em
  minus 0.4em\relax Cambridge University Press, 2008.

\bibitem{dai1995positive}
J.~G. Dai, ``On positive harris recurrence of multiclass queueing networks: a
  unified approach via fluid limit models,'' \emph{The Annals of Applied
  Probability}, pp. 49--77, 1995.

\bibitem{andrews2004scheduling}
M.~Andrews, K.~Kumaran, K.~Ramanan, A.~Stolyar, R.~Vijayakumar, and P.~Whiting,
  ``Scheduling in a queuing system with asynchronously varying service rates,''
  \emph{Probability in the Engineering and Informational Sciences}, vol.~18,
  no.~2, pp. 191--217, 2004.

\bibitem{Marg00}
C.~Maglaras, ``Discrete-review policies for scheduling stochastic networks:
  Trajectory tracking and fluid-scale asymptotic optimality,'' \emph{Annals of
  Applied Probability}, pp. 897--929, 2000.

\bibitem{shroff13}
B.~Ji, C.~Joo, and N.~Shroff, ``Throughput-optimal scheduling in multihop
  wireless networks without per-flow information,'' \emph{IEEE/ACM Transactions
  On Networking}, vol.~21, no.~2, pp. 634--647, 2013.

\bibitem{Bert15}
D.~Bertsimas, E.~Nasrabadi, and I.~C. Paschalidis, ``Robust fluid processing
  networks,'' \emph{IEEE Transactions on Automatic Control}, vol.~60, no.~3,
  pp. 715--728, 2015.

\bibitem{srikant}
B.~Li and R.~Srikant, ``Queue-proportional rate allocation with per-link
  information in multihop wireless networks,'' \emph{Queueing Systems},
  vol.~83, no. 3-4, pp. 329--359, 2016.

\bibitem{stai2016performance}
E.~Stai, S.~Papavassiliou, and J.~S. Baras, ``Performance-aware cross-layer
  design in wireless multihop networks via a weighted backpressure approach,''
  \emph{IEEE/ACM Transactions on Networking}, vol.~24, no.~1, pp. 245--258,
  2016.

\bibitem{ashok2016distributed}
A.~Krishnan K.~S. and V.~Sharma, ``A distributed algorithm for
  quality-of-service provisioning in multihop networks,'' in \emph{2017 Twenty
  Third National Conference on Communications (NCC)}, Chennai, India, Mar. 2-4,
  2017.

\bibitem{ashokFluid1}
------, ``A distributed scheduling algorithm to provide quality-of-service in
  multihop wireless networks,'' To be presented at \emph{IEEE Global
  Telecommunications Conference(GLOBECOM)}, Singapore, Dec 3-9, 2017.

\bibitem{Ber10}
D.~P. Bertsekas, ``Incremental gradient, subgradient, and proximal methods for
  convex optimization: A survey,'' \emph{Optimization for Machine Learning},
  vol. 2010, no. 1-38, p.~3, 2011.

\bibitem{Neumann}
J.~Von~Neumann, \emph{Functional operators. Volume II, The Geometry of
  orthogonal spaces}.\hskip 1em plus 0.5em minus 0.4em\relax Princeton
  University Press, 1950.

\bibitem{babyRudin}
W.~Rudin, \emph{Principles of mathematical analysis}.\hskip 1em plus 0.5em
  minus 0.4em\relax McGraw-hill New York, 1964.

\bibitem{athreya2006measure}
K.~B. Athreya and S.~N. Lahiri, \emph{Measure theory and probability
  theory}.\hskip 1em plus 0.5em minus 0.4em\relax Springer Science \& Business
  Media, 2006.

\end{thebibliography}
\bibliographystyle{IEEEtran}

\end{document}